\RequirePackage{fix-cm} %
\documentclass{sig-alternate-05-2015}

\newif\ifshort
\shortfalse %
\ifshort
\newcommand{\shortversion}{true}
\else
\newcommand{\shortversion}{false}
\fi
\newcommand{\shortver}[1]{\ifthenelse{\equal{\shortversion}{true}}{{#1}}{}}
\newcommand{\longver}[1] {\ifthenelse{\equal{\shortversion}{false}}{{#1}}{}}

\ifshort
\else
\fi

\usepackage[small,bf]{caption} 
\usepackage[labelformat=simple]{subcaption}

\usepackage{times}

\newcommand{\descr}[1]{\smallskip \noindent \textbf{#1}}

\newtheorem{theorem}{Theorem}

\usepackage{color}
\usepackage[usenames,dvipsnames]{xcolor}
\usepackage{balance}

\usepackage{algorithmic}
\usepackage[ruled, noend,lined,linesnumbered]{algorithm2e} 

\newcommand{\EMB}{\ensuremath{\mbox{$\mathcal{A}$}}}
\newcommand{\CMB}{\ensuremath{\mbox{$\mathcal{B}$}}}
\newcommand{\TMB}{\ensuremath{\mbox{$\mathcal{C}$}}}
\usepackage{url}
	\makeatletter
	\def\url@leostyle{%
	 \@ifundefined{selectfont}{\def\UrlFont{\normalsize}}%
	 {\def\UrlFont{\normalsize}}%
	}
	\makeatother	
\usepackage{cite}	%
\usepackage{graphicx}
\graphicspath{{figures/}}
\definecolor{darkgreen}{RGB}{47,109,79}
\definecolor{darkblue}{RGB}{57,79,99}
\usepackage[bookmarks=false, colorlinks=true, plainpages = false, linkcolor = darkgreen,  citecolor = darkgreen, urlcolor = darkblue, filecolor = blue]{hyperref}

\newcommand{\subparagraph}{}

\begin{document}

\CopyrightYear{2016} \setcopyright{acmcopyright}
\conferenceinfo{HotMiddlebox'16}{August 22-26, 2016, Florianopolis, Brazil}
\isbn{978-1-4503-4078-6/16/03}
\acmPrice{\$15.00}
\doi{http://dx.doi.org/10.1145/2876019.2876021} 

\ifshort
\else
\makeatletter
\def\@copyrightspace{\relax}
\makeatother
\fi

\clubpenalty=10000 
\widowpenalty=10000

\sloppy

\title{SplitBox: Toward Efficient Private Network\\Function Virtualization\ifshort\else\titlenote{\small An earlier version of this paper appeared in the Proceedings of the ACM SIGCOMM Workshop on Hot Topics in Middleboxes and Network Function Virtualization (HotMiddleBox 2016).}\fi}

\numberofauthors{6}

\author{
\alignauthor
Hassan Jameel Asghar\\
		\affaddr{Data61, CSIRO}\\
    	\affaddr{hassan.asghar@data61.csiro.au}
\alignauthor
Luca Melis\\
		\affaddr{University College London}\\
		\affaddr{luca.melis.14@ucl.ac.uk} 
\alignauthor
Cyril Soldani\\
		\affaddr{University of Li\'{e}ge}\\
		\affaddr{cyril.soldani@ulg.ac.be}
\and
\alignauthor
Emiliano De Cristofaro\\
		\affaddr{University College London}\\
		\affaddr{e.decristofaro@ucl.ac.uk}
\alignauthor
Mohamed Ali Kaafar\\
		\affaddr{Data61, CSIRO}\\
		\affaddr{dali.kaafar@data61.csiro.au}
\alignauthor
Laurent Mathy\\
		\affaddr{University of Li\'{e}ge}\\
		\affaddr{laurent.mathy@ulg.ac.be}
}

\maketitle
\begin{abstract}
This paper presents SplitBox, a scalable system for privately processing network functions that are outsourced as software processes to the cloud. Specifically, providers processing the network functions do not learn the network policies instructing how the functions are to be processed. We first propose an abstract model of a generic network function based on match-action pairs, assuming that this is processed in a distributed manner by multiple honest-but-curious providers. Then, we introduce our SplitBox system for private network function virtualization and present a proof-of-concept implementation on FastClick -- an extension of the Click modular router -- using a firewall as a use case. Our experimental results show that SplitBox achieves a throughput of over 2~Gbps with 1~kB-sized packets on average, traversing up to 60 firewall rules.
\end{abstract}

\section{Introduction}
Network function virtualization (NFV) is increasingly being adopted by organizations worldwide, 
moving network functions traditionally implemented on hardware middleboxes (MBs) -- e.g., firewalls, NAT, intrusion detection systems --
to flexible and easier to maintain software processes. Network functions can thus be executed on virtual machines (VMs), 
with cloud providers processing traffic destined to, or originating from, an enterprise network (the client) based on a set of policies governing the network functions. This, however, implies that confidential information as well as sensitive network policies (e.g., the firewall rules) are revealed to the cloud, whereas in the traditional setting, such policies would only be known to the client's network administrators. Disclosing such policies can reveal sensitive details such as  the IP addresses of hosts, the topology of the client's private network, and/or important practices~\cite{bf-firewall, mlm-firewall}. 

This motivates the need to allow processing outsourced network functions without revealing the policies: we denote this problem as {\em Private Network Function Virtualization} (PNFV), as done in~\cite{central-pnfv}. 
We argue that PNFV solutions should not only provide strong {\em security} guarantees, but also satisfy {\em compatibility} with existing infrastructures (e.g., not requiring third parties, sending/receiving traffic, take part in new protocols) as well as {\em high throughput} in order to match the quality of service expected of network functions. In practice, this precludes the use of some standard cryptographic tools as well as other approaches which we review in Section~\ref{sec:related}.

Several attempts have recently been made to support PNFV or similar functionalities~\cite{bf-firewall, mlm-firewall, central-pnfv, embark}, assuming the cloud to be honest-but-curious (i.e., the cloud processes the network functions as instructed but may try to learn the underlying policies). However, none of these simultaneously achieve
security, compatibility, and high throughput, or their coverage of network functions is limited as they are only applicable to firewall rules that either allow or drop a packet. 

Our intuition is to leverage the distributed nature of cloud VMs: rather than assuming that a single VM processes a client's network function, we distribute the functionality to several VMs residing on multiple clouds or multiple compute nodes in the same cloud. Assuming that not all VMs in the cloud are simultaneously under the control of the adversary (for instance, a \emph{passive} attacker cannot gain access to all nodes running the distributed VMs), we are able to provide a scalable and secure solution. As discussed throughout the paper, achieving this solution is not straightforward and, in the process, we overcome several challenges.

We start by presenting an abstract definition of a network function. Then, we introduce a novel system, which we name SplitBox, geared to privately and efficiently compute this abstract network function in such a way that the cloud, comprising of several middleboxes implemented as VMs, cannot learn the policies. Finally, we implement and evaluate SplitBox on a firewall test case, showing that it can achieve 
a throughput of over 2~Gbps with 1~kB-sized packets, on average, traversing up to 60 rules.

\section{Related Work}
\label{sec:related}
Khakpour and Liu~\cite{bf-firewall} present a scheme based on Bloom Filters (BFs) to privately outsource 
firewalls. Besides only considering one use case, 
their solution is not provably secure as BFs are not {\em one-way}.
\longver{Furthermore BFs inevitably introduce false positives, i.e., packets might accidentally be matched against a firewall rule.}
Privately outsourcing firewalls is also considered by Shi et al.~\cite{mlm-firewall}, who 
rely on CLT multilinear maps~\cite{mlm-integers}, which have been shown to be
insecure~\cite{mlm-cryptanalysis}.%
\longver{More specifically, the \texttt{isZero} routine of CLT maps, used in~\cite{mlm-firewall} to check whether a packet matches a policy, is not secure. Additionally note that both~\cite{bf-firewall,mlm-firewall} do not consider network functions that modify packet contents, whereas, we aim to cover a broader range of network functions including but not limited to firewalls.}
Jagadeesan et al.~\cite{hotsdn14} introduce a secure multi controller architecture for SDNs based on secure multi-party computation, which can potentially be employed for NFV. 
\longver{They provide a proof of concept implementation for identifying heavy hitters in a network consisting of two controllers. However, it}
\shortver{However, their implementation}%
takes more than 13 minutes to execute with 4096 flow table entries.
Melis et al.~\cite{central-pnfv} investigate the feasibility of provably-secure PNFV for generic network 
functions: they introduce two constructions based on fully homomorphic encryption and public-key encryption with keyword search (PEKS)~\cite{peks}, however, with high computational and communication overhead (e.g., it takes at least 250ms in their experiments to process 10 firewall rules) which makes it unfeasible for real-world deployment.

Blindbox~\cite{blindbox} considers a setting in which a sender (S) and a receiver (R) communicate via HTTPS through a middlebox (MB) which has a set of rules for packet inspection that only it knows. The MB should not be able to decrypt traffic between S and R, while S and R should not learn the rules. Although Blindbox achieves a 166Mbps throughput, 
it operates in a different setting than ours, in which R should set and know the rules (policies), while S and MB should not.
\shortver{Also, it only considers actions limited to drop, allow, or report, while we also consider modifying packet contents.}
Furthermore, the HTTPS connection setup requires around 1.5 minutes with thousands of rules, which suggests that BlindBox may not be practical for applications %
with short-lived connections.
\longver{Lin et al.~\cite{icc16} also propose a privacy-preserving deep packet filtering technique (DPF-ET) where the packet data is hidden from the network owner and the users do not learn the filtering rules. 
Compared to BlindBox, DPF-ET significantly reduces the setup overhead and requires a filtering time for each packet of $5{\mu}s$ for matching a thousand rules with a 32-bit rule length.
Both Blindbox and DPF-ET only consider middlebox actions that are limited to drop, allow or report to network administrator, without defining action as modifying packet contents (e.g., for a NAT) as is done in our paper.}

Finally, Embark~\cite{embark} enables a cloud provider to support middlebox outsourcing, such as firewalls and NATs, while maintaining confidentiality of an enterprise's network packets and policies.
\longver{Embark employs the same architecture as APLOMB~\cite{aplomb}, where the middlebox functionalities (e.g. firewall) are outsourced to the cloud by the enterprises without greatly damaging throughput, but it encrypts the traffic going to the service provider (SP) in order to protect privacy. 
To this end, Embark relies on symmetric-key encryption and introduces a novel scheme PrefixMatch used to encrypt a set of rules for a middlebox type. The encrypted rules are generated by the enterprise(s) and then provided to the SP at setup time. 
The cloud middleboxes at SP then process the encrypted traffic against the encrypted rules, and send back the produced encrypted traffic to the enterprise who, finally, performs the decryption.
When compared to Blindbox, Embark achieves better performance, as it does not require per-user-connection overhead, and broader functionality. A key difference between Embark and our solution is that we allow complex actions (on top of allow/block) to be performed on the packet without revealing them to the cloud, e.g., NAT rules. Embark can only do so in the clear.}%
\shortver{Specifically, it uses symmetric-key encryption to allow communication between enterprises and third-parties or enterprise-to-enterprise. A key difference between Embark and our solution is that we allow complex actions (besides allow/block) to be performed on the packet without revealing them to the cloud, e.g., NAT rules, while Embark can only do so in the clear.}

\section{Preliminaries}

\begin{figure}[!t]
\centering
\resizebox{1.035\columnwidth}{!}{\input{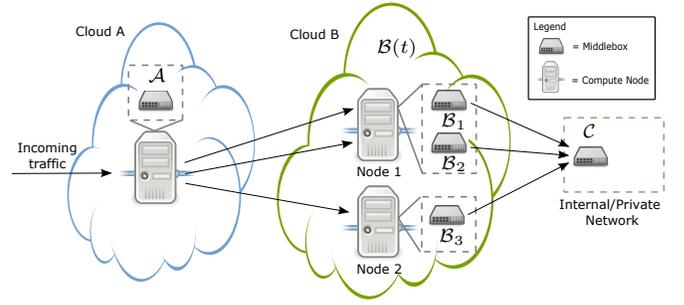}}
\caption{Our system model with Cloud A hosting MB $\EMB$ as a VM in one of its compute nodes. 
Cloud B hosts the MBs $\CMB(t)$ with $t = 3$ as VMs (not all $t$ reside on the same compute node). 
The client MB $\TMB$ resides at the edge of the client's internal network. $\EMB$ and $\CMB(t)$ collaboratively compute network functions for the client.}
\label{fig:pnfv-overview}
\shortver{\vspace{-0.2cm}}
\end{figure}  

\longver{\subsection{System and Trust Model}}
\shortver{\descr{System and Trust Model.}}
Figure~\ref{fig:pnfv-overview} illustrates our PNFV model, consisting of two types of cloud middleboxes (MBs): an \emph{entry} MB $\EMB$ and $t \ge 2$ cloud MBs $\CMB(t)$, which collaboratively compute a network function on behalf of a client. The client has its own MB, denoted $\TMB$, at the edge of its internal network. $\EMB$ receives an incoming packet, does some computations on it, ``splits'' the result into $t$ parts, and forwards part $j$ to $\CMB_j\in\CMB(t)$. $\CMB_j$ performs local computations and forwards its part to $\TMB$, which reconstructs the network function's final result. 
\longver{There is also a direct link between $\EMB$ and $\TMB$.}

\descr{Assumptions.} We assume an honest-but-curious adversary which can corrupt\footnote{\small The adversary may change the behavior of a MB from honest to honest-but-curious.} either $\EMB$ or up to $t-1$ MBs from $\CMB(t)$, and it cannot corrupt $\EMB$ and any MB in $\CMB(t)$ simultaneously. In practice, one can assume $\EMB$ to be running on a different cloud provider than $\CMB(t)$ and that not all MBs in $\CMB(t)$ reside on the same node.
\longver{Since $\TMB$ is the client's MB, we do not assume it to be adversarial.}

\longver{\subsection{Network Functions}}
\shortver{\descr{Network Functions.}}
We define a packet $x$ as a binary string of arbitrary length, i.e., $x \in \{0, 1\}^*$. Our network functions will be applicable to the first $n$ bits of $x$.
\longver{
If $|x| < n$, $x$ is prefixed with zeros to make it of length $n$.
}
A \textit{matching} function is a boolean function $m : \{0, 1\}^{n} \rightarrow \{0, 1\}$.
Its complement, i.e., the function $1 - m$, is denoted by $\overline{m}$. 
An \textit{action} function is a transformation $a :  \{0, 1\}^{n} \rightarrow \{0, 1\}^{n}$. 
$m(x)$ (resp., $a(x)$) denote evaluating $m$ (resp., $a$) on the substring $x(1, n)$ (i.e., the first $n$ bits of $x$). If $|x| > n$, $a$ keeps the part $x(n + 1, *)$ of $x$ unaltered. We also define the identity action function $I(x) = x$. 

Let $M$ and $A$ be finite sets of matching and action functions, with $I \in A$. A \textit{network} function $\psi = (M, A)$ is a binary tree with edge set $M$ and node set $A$ such that each node is an action function $a \in A$ and each edge is either a matching function $m \in M$ or a complement $\overline{m}$ of a matching function $m \in M$. A node is either a leaf node or a parent node. A parent node has two child nodes. The left child node is the identity action function $I$. The edge connecting the right child node is a matching function $m \in M$, whereas the edge connecting the left child node is its complement $\overline{m}$. The root node is the identity action function $I$. \longver{Examples of network functions are in Figure~\ref{fig:nf-binary-trees}.}
Clearly, there exists a binary relation from $M$ to $A$, such that for each $(m, a)$ from this relation there exists a parent node in $\psi$ such that
the left child is connected via the edge $\overline{m}$ and the right child via the edge $m$, and the right child is $a$.

We call each pair $(m, a)$ in $\psi$ a \emph{policy}.
\longver{A policy can also be represented as a subtree of $\psi$ as shown in Figure~\ref{subfig:policy-binary-tree}.} 
Policies serve as building blocks of a network function.
The set of policies of $\psi$ is the set of \emph{distinct} policies $(m, a)$ in $\psi$. A network function is evaluated on input $x \in \{0, 1\}^*$, denoted $\psi(x)$, using Algorithm~\ref{alg:traversal}.
\longver{Note that the reason to create a separate writeable copy ${x}_\mathtt{w}$ of $x$ is to ensure that the matching functions are applied on the ``unmodified'' $x$, i.e., $x_\mathtt{r}$, and not on ${x}_\mathtt{w}$ which is modified by the action functions. When a leaf node is entered, we say that the network function has terminated. %
}
Figure~\ref{subfig:nf-binary-tree-unbalanced} shows a network function with $k$ distinct policies: whenever a match is found, the corresponding action is performed and the function terminates. The function in Figure~\ref{subfig:nf-binary-tree-balanced-k3} has $3$ distinct policies, $(m_1, a_1), (m_2, a_2)$ and $(m_3, a_3)$, and $(m_2, a_2)$ is repeated twice. This function does not terminate immediately after a match has been found (e.g., path $m_1m_2$). Since $a \circ I = I \circ a = a$, we can easily ``plug'' individual policy trees %
to construct more complex network functions.

\begin{algorithm}[ttt]
\caption{\texttt{Traversal}}
\label{alg:traversal}
\SetAlgoLined
\SetCommentSty{mycommfont}
\SetAlCapSkip{1em}
\DontPrintSemicolon{}
\SetKwInOut{Input}{Input}
\let\oldnl\nl%
\newcommand{\nonl}{\renewcommand{\nl}{\let\nl\oldnl}}%
\Input{Packet $x$, network function $\psi$.}
Make a read-only copy $x_\mathtt{r}$ and a writeable copy ${x}_{\mathtt{w}}$ of $x$. \;
Start from the root node.\;
Compute ${x}_\mathtt{w} \leftarrow a({x}_\mathtt{w})$, where $a$ is the current node.\label{algstep:loop-step}\;
\If{the current node is a leaf node}{
	output ${x}_\mathtt{w}$ and stop.
}
\Else{
	Compute $m(x_\mathtt{r})$, where $m$ is the right hand side edge.\;
	\If{$m(x_\mathtt{r}) = 1$}{
		Move to the right child node.
	}
	\Else{
		Move to the left child node.
	}
}
Go to step~\ref{algstep:loop-step}. %
\end{algorithm}

\shortver{
\begin{figure}[!t]
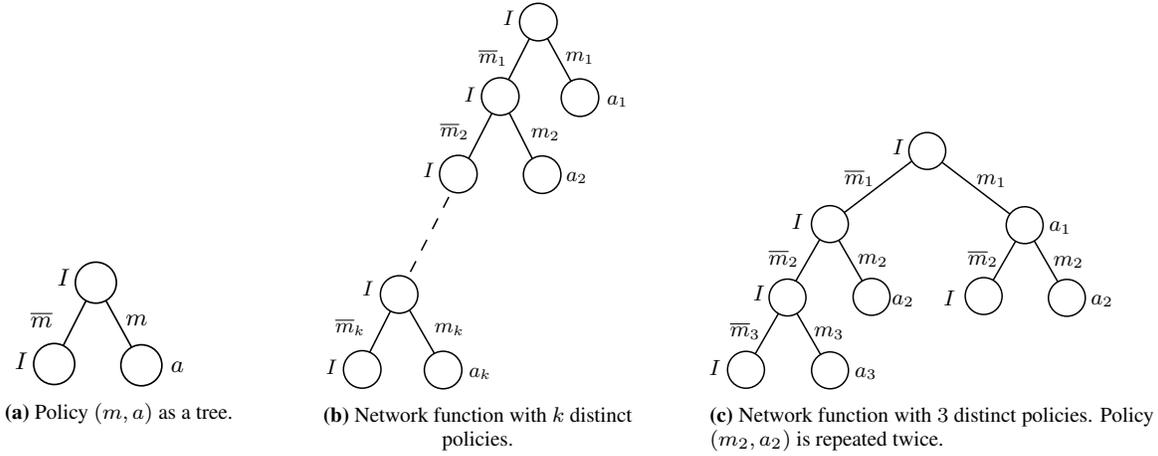

\centering
  \begin{subfigure}[b]{0.38\columnwidth}\centering
 		 \resizebox{1\textwidth}{!}{\input{figures/nf-binary-tree-unbalanced.tex}}
		\caption{Network function with $k$ distinct policies. }
		\label{subfig:nf-binary-tree-unbalanced}
  \end{subfigure}
~
  \begin{subfigure}[b]{0.5535\columnwidth}\centering
 		 \resizebox{1\textwidth}{!}{\input{figures/nf-binary-tree-balanced-k3.tex}}
		\caption{Network function with $3$ distinct policies. Policy $(m_2, a_2)$ is repeated twice.}
		\label{subfig:nf-binary-tree-balanced-k3}
  \end{subfigure}
  \vspace{-0.2cm}
\caption{Network functions as binary trees.}
\label{fig:nf-binary-trees}
\vspace{-0.1cm}
\end{figure}}

\longver{
\begin{figure*}[!t]
\centering
	\begin{subfigure}[b]{0.4\columnwidth}\centering
 		\resizebox{0.8\columnwidth}{!}{\input{figures/policy-binary-tree.tex}}
		\caption{Policy $(m,a)$ as a tree.\\[3ex] }
		\label{subfig:policy-binary-tree}
  \end{subfigure}
  ~~~~~
  \begin{subfigure}[b]{0.6\columnwidth}\centering
 		 \resizebox{0.8\textwidth}{!}{\input{figures/nf-binary-tree-unbalanced.tex}}
		\caption{Network function with $k$ distinct\\ policies. }
		\label{subfig:nf-binary-tree-unbalanced}
  \end{subfigure}
  ~~~~~
    \begin{subfigure}[b]{0.7\columnwidth}\centering
 		 \resizebox{1\textwidth}{!}{\input{figures/nf-binary-tree-balanced-k3.tex}}
		\caption{Network function with $3$ distinct policies. Policy $(m_2, a_2)$ is repeated twice.}
		\label{subfig:nf-binary-tree-balanced-k3}
  \end{subfigure}
  \vspace{-0.2cm}
\caption{Network functions as binary trees.}
\label{fig:nf-binary-trees}
\end{figure*}
}  

\descr{Coverage.} Our abstract definition of network functions captures many network functions used in practice. These include firewalls, NAT and load balancers. Such functions usually perform a matching step to inspect some parts of a packet and modify contents of the packet subsequently. In the case of firewalls, modifications may also include dropping a packet.

\descr{Branching and chaining.} Our definitions support branching, i.e., network functions that do not necessarily apply all policies on a packet. This is achieved by including multiple exit points, i.e., leaf nodes.
\longver{In this sense, our definition of network functions is richer than the one proposed in~\cite{central-pnfv} which does not allow composition of policies with multiple exit points.}
Definitions also support \emph{chaining}, e.g., $\psi_1$'s output is $\psi_2$'s input, however, in our proposed privacy-preserving solution chaining is not possible, since outputs of the MBs in $\CMB(t)$ need to be combined to reconstruct a transformed packet. 
\longver{For chaining to work, network function $\psi_2$ needs to know the output of network function $\psi_1$. However, if $\psi_2$ only needs the original input $x$, instead of the overwritten copy $x_\mathtt{w}$, network function chaining can work by giving $\psi_2$ an auxiliary input, i.e., the share resulting from network function $\psi_1$, on which it can apply its own actions.}

\longver{\subsection{Policies}}
\shortver{\descr{Policies.}}
We restrict $m$ to substring matching and $a$ to be substring substitution. We also introduce the \emph{don't care bit} denoted by~$*$ in our alphabet. 
Given strings $x \in \{0, 1\}^n$ and $y \in \{0, 1, * \}^n$, we say $x = y$ if $x(i) = y(i)$ for all $i \in [n]$ such that $y(i) \ne *$. 
\longver{In other words, if the two strings match at every position except for the don't care positions we consider the two strings to be equal.}
Given $x \in \{0, 1\}^*$, matching function $m$ is defined as \shortver{$m(x)=1$ if $x(1, n) = \mu$ and $0$ otherwise,}
\longver{
\begin{equation}
\label{eq:matching}
m(x) = \begin{cases}
				1, & \text{if } x(1, n) = \mu\\
				0, & \text{otherwise}
				\end{cases},
\end{equation}}
where $\mu \in \{0, 1, *\}^n$. We call $\mu$ the \textit{match} of $m$. 
\longver{To define the action function, we introduce substring replacement.}
Given $x \in \{0, 1\}^n$ and $z \in \{0, 1, *\}^n$, $x \leftarrow z$ represents replacing each $x(i)$ with $z(i)$ if $z(i) \ne *$, and leaving $x(i)$ as is if $z(i) = *$, for all $i \in [n]$. Given $x \in \{0, 1\}^*$, the action function $a$ is defined as %
\longver{
\begin{equation}
\label{eq:action}
a(x) = x(1, n) \leftarrow \alpha, 
\end{equation}}
\shortver{$a(x) = x(1, n) \leftarrow \alpha$,}
where $\alpha \in \{0, 1, *\}^n$. We call $\alpha$ the \textit{action} of $a$. \shortver{For the identity function $I$, $\alpha = *^n$.}%
\longver{With this definition, the identity action function $I$ is $I(x) = x(1, n) \leftarrow \alpha$, where $\alpha = *^n$. }

\descr{Definitions.} Throughout the rest of the paper, we use the following definitions: let $z \in \{0, 1, * \}^n$, the \emph{projection} of $z$, denoted $\pi_z$, is a string $\in \{0, 1\}^n$, s.t. $\pi_z(i) = 1$ if $z(i) \in \{0, 1\}$ and $\pi_z(i) = 0$ if $z(i) = *$. The \emph{masking} of a 
\longver{
$x \in \{0, 1, *\}^n$
}
\shortver{
$x \in \{0, 1\}^n$
} 
using $\pi_z \in \{0, 1\}^n$, denoted $\omega(\pi_z, x)$, returns $x'$ s.t. $x'(i) = x(i)$ if $\pi_z(i) = 1$ and $x'(i) = 0$ if $\pi_z(i) = 0$.
\longver{Although we have broadly defined $\omega(\pi_z, x)$ for an $x \in \{0, 1, *\}^n$, we use it exclusively for an $x \in \{0, 1\}^n$. 
}
$\mathbb{H}: \{0, 1\}^n \rightarrow \{0, 1\}^q$ denotes a cryptographic hash function; $\oplus$ denotes bitwise XOR. The Hamming weight of a string $x \in \{0, 1\}^n$ is $\text{wt}(x)$. Finally,
$x \leftarrow_{\$} \{0, 1\}^n$ means sampling a binary string of length $n$ uniformly at random. 

\section{Introducing SplitBox}\label{sec:solution}

\longver{\subsection{Privacy Requirements}}
\shortver{\descr{Privacy Requirements.}}
We start by describing an \emph{ideal} setting in which a trusted third party, $\mathcal{T}$, computes a network function $\psi$ for the client. Upon receiving a packet $x$, $\EMB$ forwards it to $\mathcal{T}$, which provides the result of $\psi(x)$ to $\TMB$. Here $\EMB$ learns $x$ but not $\psi(x)$ and $\CMB(t)$ neither $x$ nor $\psi(x)$. In this section, we introduce our private NFV solution, SplitBox, aiming to simulate this ideal setting. However, we fall slightly short in that the MBs $\CMB(t)$ learn the projection $\pi_\mu$ and the output $m(x)$ for each $m \in M$, however, they do not learn the match $\mu$ for any $m \in M$ beyond what is learnable from $\pi_\mu$. Although this could reveal information such as which field of the packet the current matching function corresponds to, we do not consider it to be a strong limitation since this might be obvious from the type of NFV considered anyway. For example, if it is a firewall, then it is common knowledge that the fields it operates on will include IP address fields.

\longver{\subsection{The System}}
\descr{Design Aims.}
We consider the following design aims, i.e., the solution should: (a) 
be secure; (b) be computationally fast; (c) 
limit MB-to-MB communication complexity.

\begin{figure}[!tb]
\centering
\resizebox{1\columnwidth}{!}{
\begingroup%
 \makeatletter%
 \providecommand\color[2][]{%
  \errmessage{(Inkscape) Color is used for the text in Inkscape, but the package 'color.sty' is not loaded}%
  \renewcommand\color[2][]{}%
 }%
 \providecommand\transparent[1]{%
  \errmessage{(Inkscape) Transparency is used (non-zero) for the text in Inkscape, but the package 'transparent.sty' is not loaded}%
  \renewcommand\transparent[1]{}%
 }%
 \providecommand\rotatebox[2]{#2}%
 \ifx\svgwidth\undefined%
  \setlength{\unitlength}{200.85478516bp}%
  \ifx\svgscale\undefined%
   \relax%
  \else%
   \setlength{\unitlength}{\unitlength * \real{\svgscale}}%
  \fi%
 \else%
  \setlength{\unitlength}{\svgwidth}%
 \fi%
 \global\let\svgwidth\undefined%
 \global\let\svgscale\undefined%
 \makeatother%
 \begin{picture}(1,0.7)%
  \put(0,0){\includegraphics[scale=0.5,width=\unitlength]{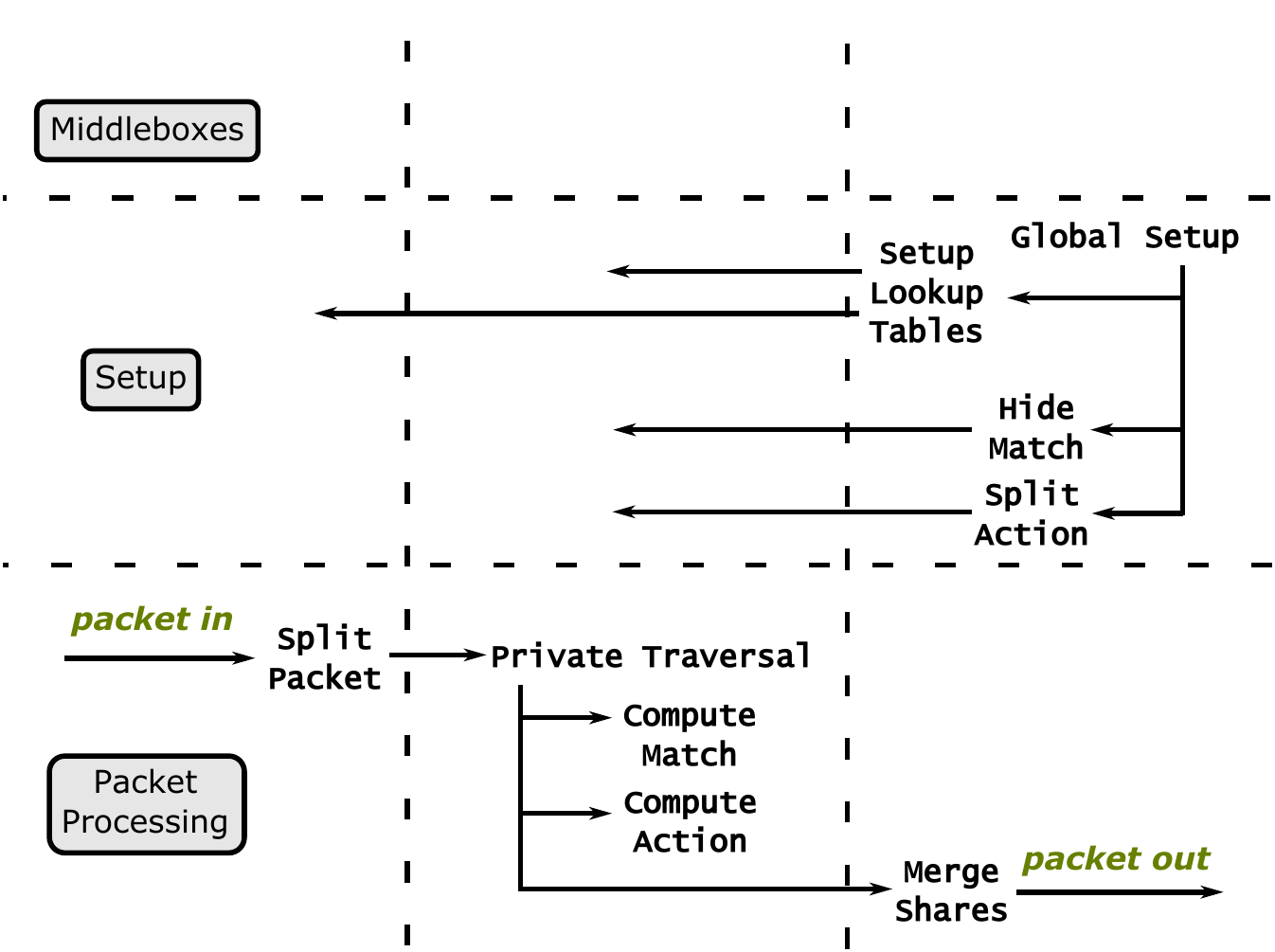}}%
  \put(0.23,0.635){\color[rgb]{0,0,0}\makebox(0,0)[lb]{\smash{\scriptsize $\EMB$}}}%
  \put(0.47,0.635){\color[rgb]{0,0,0}\makebox(0,0)[lb]{\smash{\scriptsize $\CMB(t)$}}}%
  \put(0.82,0.635){\color[rgb]{0,0,0}\makebox(0,0)[lb]{\smash{\scriptsize $\TMB$}}}%
 \end{picture}%
\endgroup%
}
\caption{Breakdown of algorithms executed by each MB in SplitBox.}
\vspace{-0.1cm}
\label{fig:algo-breakdown}
\end{figure} 

\descr{High-Level Overview.}
In a nutshell, if we assume that $\psi$ includes a single policy $(m, a)$, our strategy to hide $m$ 
is to let $\TMB$ \emph{blind} $\mu$ by XORing it with a random binary string $s$ and sending the hash of the result to each MB in $\CMB(t)$; whereas, to hide $a$, $\TMB$ computes $t$ \emph{shares} of the action $\alpha$ using a $t$-out-of-$t$ secret sharing scheme and sends share $j$ to $\CMB_j$. In addition, $\EMB$ encrypts the contents of a packet $x$ by XORing it with the blind $s$, and sends it to the MBs in $\CMB(t)$, which can then compute matches and actions on this encrypted packet. 
We present the details of SplitBox using a set of algorithms, grouped based on the MB executing them. Figure~\ref{fig:algo-breakdown} shows a high-level overview of all the algorithms computed by each MB. We assume $\psi_{\text{priv}}$ to be the private version of the network function $\psi$ whose matching and action functions are replaced by unique identifiers. 

\descr{Middlebox \texorpdfstring{$\TMB$}{C)}.}
The initial setup is performed by $\TMB$ via Algorithm~\ref{alg:glob_setup}.
This includes creating lookup tables (Algorithm~\ref{alg:lookup_setup}), hiding
the matching functions (Algorithm~\ref{alg:hide_match}), and splitting the
action functions (Algorithm~\ref{alg:split_action}). There are two lookup
tables in Algorithm~\ref{alg:lookup_setup}: $S$ for $\mathcal{A}$ and
$\tilde{S}$ for $\CMB(t)$. Table $S$ contains $l$ ``blinds'' which are random
binary strings used to encrypt a packet by XORing. For each blind $s \in S$ and
for each $m \in M$, the portion of the blind corresponding to the projection of
the match $\mu$ is extracted and then XORed with $\mu$. Finally this value is
hashed using $\mathbb{H}$ and stored in the corresponding row of $\tilde{S}$.
The \texttt{Hide Match} algorithm simply sends the projection $\pi_\mu$ of each
match $\mu$ to $\CMB(t)$. This tells $\CMB(t)$ which locations of the incoming
packet are relevant for the current match. The \texttt{Split Action} algorithm
computes $t$ shares of the action $\alpha$ and action projection $\pi_\alpha$,
for each $a \in A$ and sends them to $\CMB(t)$. $\TMB$ uses one more algorithm,
Algorithm~\ref{alg:merge_shares} to reconstruct the transformed packet. This
algorithm XORs the cumulative action shares $\alpha'_j$ and cumulative action
projection shares $\beta'_j$ from $\CMB_j$ to compute the final action
$\alpha'$ and action projection $\beta'$. It also XORs the
encrypted packet received from $\mathcal{A}$ with the current blind $s$ in the lookup table $S$, in order to
reconstruct the final packet. Note that we have modelled dropping a packet as setting $x(1, n)$ to $0^n$. 

\begin{algorithm}[t]
\caption{\texttt{Global Setup} ($\TMB$)}
\label{alg:glob_setup}
\SetAlgoLined
\SetCommentSty{mycommfont}
\SetAlCapSkip{1em}
\DontPrintSemicolon{}
\SetKwInOut{Input}{Input}
\let\oldnl\nl%
\newcommand{\nonl}{\renewcommand{\nl}{\let\nl\oldnl}}%
\Input{Parameters $n$ and $l$, network function $\psi = (M, A)$.}
\For{$j = 1$ \KwTo $t$}{
	Send $\psi_{\text{priv}}$ to $\CMB_j$. \;
}
Run \texttt{Setup Lookup Tables} with parameter $l$, $M$.\;
\For{each $m \in M$}{
	Run \texttt{Hide Match} algorithm.\;
}
\For{each $a \in A$}{
	Run \texttt{Split Action} algorithm.\;
}
\end{algorithm}

\begin{algorithm}[t]
\caption{\texttt{Setup Lookup Tables} ($\TMB$)}
\label{alg:lookup_setup}
\SetAlgoLined
\SetCommentSty{mycommfont}
\SetAlCapSkip{1em}
\DontPrintSemicolon{}
\SetKwInOut{Input}{Input}
\let\oldnl\nl%
\newcommand{\nonl}{\renewcommand{\nl}{\let\nl\oldnl}}%
\Input{Parameter $l$, set $M$.}
Initialize empty table $S$ with $l$ cells.\;
Initialize empty table $\tilde{S}$ with $l\times |M|$ cells.\;
\For{$i = 1$ \KwTo $l$}{
	Sample $s_i \leftarrow_{\$} \{0, 1\}^n$.\;
	Insert $s_{i}$ in cell $i$ of $S$.\;
	\For{$j = 1$ \KwTo $| M |$}{
				Compute $\tilde{s}_{i, j} = \omega(\pi_{\mu_j}, s_i)$, where $\mu_j$ is the match of $m_j$.\;
				Compute $\mathbb{H}(\mu_j \oplus \tilde{s}_{i, j})$.\;
				Insert $\mathbb{H}(\mu_j \oplus \tilde{s}_{i, j})$ in cell $(i, j)$ of $\tilde{S}$.\;
	}
}
Send $S$ to $\EMB$.\;
Send $\tilde{S}$ to $\CMB(t)$. \;
\end{algorithm}

\begin{algorithm}[t] %
\caption{\texttt{Hide Match} ($\TMB$)}
\label{alg:hide_match}
\SetAlgoLined
\SetCommentSty{mycommfont}
\SetAlCapSkip{1em}
\DontPrintSemicolon{}
\SetKwInOut{Input}{Input}
\let\oldnl\nl%
\newcommand{\nonl}{\renewcommand{\nl}{\let\nl\oldnl}}%
\Input{Matching function $m \in M$ with match $\mu$.}
Send $\pi_{\mu}$ to $\CMB(t)$.\;
\end{algorithm}

\begin{algorithm}[t] %
\caption{\texttt{Split Action} ($\TMB$)}
\label{alg:split_action}
\SetAlgoLined
\SetCommentSty{mycommfont}
\SetAlCapSkip{1em}
\DontPrintSemicolon{}
\SetKwInOut{Input}{Input}
\let\oldnl\nl%
\newcommand{\nonl}{\renewcommand{\nl}{\let\nl\oldnl}}%
\Input{Action function $a \in A$ with action $\alpha$.}
Sample $\alpha_1, \alpha_2, \ldots, \alpha_{t-1} \leftarrow_{\$} \{0, 1\}^n$.\;
Let $\tilde{\alpha} = \omega(\pi_{\alpha}, \alpha)$. Compute $\alpha_t = \tilde{\alpha} \oplus \alpha_1 \oplus \cdots \oplus \alpha_{t-1}$.\;
Sample $\beta_1, \beta_2, \ldots, \beta_{t-1} \leftarrow_{\$} \{0, 1\}^n$.\;
Compute $\beta_t = \pi_{\alpha} \oplus \beta_1 \oplus \cdots \oplus \beta_{t-1}$.\;
\For{$j = 1$ \KwTo $t$}{
	Give $\alpha_j, \beta_j$ to $\CMB_j$.\;
}
\end{algorithm}

\begin{algorithm}[t] %
\caption{\texttt{Merge Shares} ($\TMB$)}
\label{alg:merge_shares}
\SetAlgoLined
\SetCommentSty{mycommfont}
\SetAlCapSkip{1em}
\DontPrintSemicolon{}
\SetKwInOut{Input}{Input}
\let\oldnl\nl%
\newcommand{\nonl}{\renewcommand{\nl}{\let\nl\oldnl}}%
\Input{Index $i$, packet copy ${x}_{\mathtt{w}}$, $\alpha'_j$ and $\beta'_j$ from $\CMB_j$ for $j \in [t]$.}
Compute $\alpha' \leftarrow \alpha'_1 \oplus \cdots \oplus \alpha'_t$.\;
Compute $\beta' \leftarrow \beta'_1 \oplus \cdots \oplus \beta'_t$.\;
Compute $x \leftarrow {x}_{\mathtt{w}} \oplus s_i$, where $s_i\in S$.\;
\For{$i = 1$ \KwTo $n$}{
	\If{$\beta'(i) = 1$}{$x(i) \leftarrow \alpha'(i)$}
}
\If{$x(1, n) = 0^{n}$}{Drop $x$.}
\Else{
	Forward $x$.
}
\end{algorithm}

\descr{Middlebox \texorpdfstring{$\EMB$}{A)}.} 
This MB only runs Algorithm~\ref{alg:split_packet}, which maintains a counter initially set to $0$ and incremented every time a new packet $x$ arrives. The value of the counter corresponds to a blind in the lookup table $S$. Therefore its range is $[l]$ (barring the initial value of $0$). The algorithm makes two copies of an incoming packet $x$, $x_\mathtt{r}$ (read-only copy) for matching to be sent to $\CMB(t)$, and ${x}_\mathtt{w}$ (writeable copy) for action functions to be sent to $\TMB$. Both $x_\mathtt{r}$ and $x_\mathtt{w}$ are XORed with the blind in $S$ corresponding to the counter. The current counter value is also given to $\CMB(t)$ and $\TMB$. 

\begin{algorithm}[t] %
\caption{\texttt{Split Packet} ($\EMB$)}
\label{alg:split_packet}
\SetAlgoLined
\SetCommentSty{mycommfont}
\SetAlCapSkip{1em}
\DontPrintSemicolon{}
\SetKwInOut{Input}{Input}
\let\oldnl\nl%
\newcommand{\nonl}{\renewcommand{\nl}{\let\nl\oldnl}}%
\Input{Packet $x$, lookup table $S$.}
Get the index $i \in [l]$ corresponding to the current value of the counter.\;
Let $x_{\mathtt{w}} \leftarrow x \oplus s_i$ (writeable copy), where $s_i\in S$.\;
Compute $x_\mathtt{r} \leftarrow x(1, n) \oplus s_i$ (read-only copy), where $s_i \in S$. \;
\For{$j = 1$ \KwTo $t$}{
	Send $x_\mathtt{r}$, $i$ to $\CMB_j$.\;
}
Send $x_\mathtt{w}$, $i$ to $\TMB$.\;
\end{algorithm}

\descr{Middleboxes \texorpdfstring{$\CMB(t)$}{B(t)}.}
Each MB $\CMB_j$ performs a private version of the \texttt{Traversal} algorithm as shown in Algorithm~\ref{alg:priv_traversal}. $\CMB_j$ first initializes cumulative action strings $\alpha'_j$ and cumulative action projection strings $\beta'_j$ as strings of all zeros. Within the \texttt{Private Traversal} algorithm, $\CMB_j$ executes the action functions using Algorithm~\ref{alg:compute_action} and matching functions using Algorithm~\ref{alg:compute_match}. The \texttt{Compute Action} algorithm essentially updates $\alpha'_j$ and $\beta'_j$ by XORing with the action share and action projection share of the current action. The \texttt{Compute Match} algorithm uses the read-only copy $x_\mathtt{r}$. It extracts the bits of $x_\mathtt{r}$ corresponding to the current match projection $\pi_\mu$. It then looks up the counter value $i$ (sent by $\EMB$) and the index of the matching function in the lookup table $\tilde{S}$ and extracts the hashed match. This is then compared with the hash of the relevant bits of $x_\mathtt{r}$. 

\begin{algorithm}[h] %
\caption{\texttt{Private Traversal} ($\CMB(t)$)}
\label{alg:priv_traversal}
\SetAlgoLined
\SetCommentSty{mycommfont}
\SetAlCapSkip{1em}
\DontPrintSemicolon{}
\SetKwInOut{Input}{Input}
\let\oldnl\nl%
\newcommand{\nonl}{\renewcommand{\nl}{\let\nl\oldnl}}%
\Input{Index $i$, read-only copy $x_\mathtt{r}$, network function $\psi_{\text{priv}}$.}
Initialize empty strings $\alpha'_j \leftarrow 0^n$ and $\beta'_j \leftarrow 0^n$.\;
Start from the root node. \;
Update $\alpha'_j$ and $\beta'_j$ by running the \texttt{Compute Action} algorithm on the current node $a$.\label{algstep:loop-step-2}\;
\If{the current node is a leaf node}{
	Send $i$, $\alpha'_j$ and $\beta'_j$ to party $\TMB$ and stop.
}
\Else{
	Run \texttt{Compute Match} algorithm on $i$, $m$ and $x_\mathtt{r}$, where $m$ is the right hand side edge.\;
	\If{\emph{\texttt{Compute Match}} outputs 1}{
		Go to the right child node.
	}
	\Else{
		Go to the left child node.
	}
}
Go to step~\ref{algstep:loop-step-2}. \;
\end{algorithm}

\begin{algorithm}[h] %
\caption{\texttt{Compute Action} ($\CMB(t)$)}
\label{alg:compute_action}
\SetAlgoLined
\SetCommentSty{mycommfont}
\SetAlCapSkip{1em}
\DontPrintSemicolon{}
\SetKwInOut{Input}{Input}
\let\oldnl\nl%
\newcommand{\nonl}{\renewcommand{\nl}{\let\nl\oldnl}}%
\Input{Pair of cumulative action and cumulative action projection shares $(\alpha'_j, \beta'_j)$ of $\CMB_j$, pair of action and action projection shares $(\alpha_j, \beta_j)$ of action function $a \in A$ of $\CMB_j$.}
Compute $\alpha'_j \leftarrow \alpha'_j \oplus \alpha_j$.\;
Compute $\beta'_j \leftarrow \beta'_j \oplus \beta_j$.\;
Output $\alpha'_j$, $\beta'_j$.\;
\end{algorithm}

\begin{algorithm}[h] %
\caption{\texttt{Compute Match} ($\CMB(t)$)}
\label{alg:compute_match}
\SetAlgoLined
\SetCommentSty{mycommfont}
\SetAlCapSkip{1em}
\DontPrintSemicolon{}
\SetKwInOut{Input}{Input}
\let\oldnl\nl%
\newcommand{\nonl}{\renewcommand{\nl}{\let\nl\oldnl}}%
\Input{Read-only copy $x_\mathtt{r}$, index $i \in [l]$, lookup table $\tilde{S}$, index $j \in [|M|]$ of $m_j \in M$ with match $\mu_j$.}
Lookup table $\tilde{S}$ at index $(i, j)$ to obtain $\mathbb{H}(\tilde{s}_{i, j})$.\;
Extract $\tilde{x}_\mathtt{r} \leftarrow \omega(\pi_{\mu_j}, x_\mathtt{r})$.\;
Compute $\mathbb{H}(\tilde{x}_\mathtt{r})$.\;
\If(\tcp*[f]{$m(x) = 1$}){$\mathbb{H}(\tilde{x}_\mathtt{r}) = \mathbb{H}(\mu_j \oplus \tilde{s}_{i, j})$}{Output 1.}
\Else(\tcp*[f]{$m(x) = 0$}){Output 0.}
\end{algorithm}

\section{Analysis}
\longver{\subsection{Correctness}}
\shortver{\descr{Correctness.}}
Given $\psi = (M, A)$, for a matching function $m \in M$, as long as $m$ can be represented as substring matching, SplitBox correctly computes the match. That is, if $m$ is an equality test or range test for powers of 2 in binary (e.g., IP addresses in the range $\mathtt{127.*.*.32}$ to $\mathtt{127.*.*.64}$), then it can be successfully computed by SplitBox. Our model also allows for arbitrary ranges by dividing $m$ into smaller matches that check equality matching of individual bits. However, such a representation can potentially make $\psi$ very large. 
We can correctly compute action functions as long as they satisfy two properties: (a) they are applied to the initial packet $x$ only, and not on its transformed versions; (b) any two action projections $\beta_i$ and $\beta_j$ do not overlap on their non-zero bits. Note that this does not restrict the number of times the identity function $I$ can be applied, as its action projection is $0^n$. 

\longver{\subsection{Security}\label{sub:security}}%
\shortver{\descr{Security.}}%
\longver{The proof of security of our construction is shown in Appendix~\ref{app:sec-proofs}. Here, we}%
\shortver{While we refer to the full version of the paper for the security proofs~\cite{full}, here we}
mention two important points: if SplitBox is used for match projections whose Hamming weight is low, then the $\CMB(t)$ can brute-force $\mathbb{H}$ to find its pre-image. This reveals $\mu \oplus s$ for some blind $s$, which allows the adversary to learn more than simply looking at the output of $m$. Namely, if $m(x) = 0$, the adversary learns which relevant bits of an incoming packet $x$ do not match with the stored match. 
\longver{This is the reason why we use the hash function $\mathbb{H}$. It does not allow $\CMB(t)$ to learn more than the output of $m$. }
The second point relates to the length of the look-up table $l$: ideally $l$ should be large enough so that the same blind is not re-used before a long period of time. However, high throughput would require a prohibitively large value of $l$. Therefore, we propose the following mitigation strategy: with probability $0 < 1 - \rho < 1$, $\EMB$, sends a uniform random string from $\{0, 1\}^n$ (dummy packet), rather than the next packet in the queue. 
\longver{Thus, any middlebox in $\CMB(t)$, that attempts to compare two packets using the same blind (according to the value of the counter $i \in [l]$) does not know for certain whether the result corresponds to two actual packets (the probability is $ \rho^2$) or not. The downside is that this reduces the (effective) throughput by a factor of $\rho$. Nevertheless, with this strategy we can use a feasible value of $l$. Of course $\EMB$ has to indicate to $\TMB$ which packet is a dummy packet. This can be done by sending a bit through $\CMB(t)$ to $\TMB$ by once again using a $t$-out-of-$t$ secret sharing scheme (XORing with random bits).}
\shortver{Thus, any MB in $\CMB(t)$ does not know if the two packets corresponding to the same blind are two actual packets (the probability is $ \rho^2$). The downside is that this reduces the (effective) throughput by a factor of $\rho$. $\EMB$ can indicate to $\TMB$ if the current packet is a dummy packet by sending a bit through $\CMB(t)$ to $\TMB$ using a $t$-out-of-$t$ secret sharing scheme (XORing with random bits).}

\section{Implementation}
In this section, we discuss our proof-of-concept implementation of
SplitBox inside
FastClick~\cite{barbette2015fast}, an extension of the Click modular
router~\cite{kohler2000click} which provides fast user-space packet I/O and easy
configuration via automatic handling of multi-threading and multiple
hardware queues. We also use Intel DPDK~\cite{intelDPDK} as the underlying
packet I/O framework. We implemented three main FastClick elements: element 
\texttt{Entry} corresponding to MB $\EMB$, \texttt{Processor} corresponding 
to MBs $\CMB$, and \texttt{Client} to $\TMB$. \texttt{Client}
implements the \texttt{Merge Shares} algorithm. 
\longver{An element \texttt{Entry}
corresponds to MB $\EMB$. This element is responsible for the
\texttt{Split Packet} algorithm. An element \texttt{Processor} is used for the
$\CMB_j$ MBs. It is responsible for the \texttt{Private Traversal},
\texttt{Compute Match} and \texttt{Compute Action} algorithms. Finally, an
\texttt{Client} element corresponds to $\TMB$ and is responsible for the
\texttt{Merge Shares} algorithm.}
The other algorithms of $\TMB$ are executed outside the FastClick
elements, and used to configure the above three elements.
The hash function $\mathbb{H}$ is implemented using OpenSSL's SHA-1, aiming to
achieve a compromise between security, digest length, and computation
\shortver{speed, as}
\longver{speed. While
faster hashing functions are available, they are not cryptographic hash functions, 
thus they might be invertible and/or lead to
larger amount of collisions. On the other hand, we do not want hash functions which have very large message digests (leading to overly large lookup
tables), or which are more computationally expensive (as discussed in Section
\ref{sec:evaluation}, the hashing speed is an important factor of the performance of our
solution).
}\shortver{hash functions which have larger message digests will lead to overly large lookup tables.}
\texttt{Client} uses a circular buffer to collect packet shares until
all have been received and the final packet can be reconstructed.
For communication between our elements, we use UDP packets: UDP and L2
processing relies on standard Click elements
such as \texttt{UDPIPEncap}. Finally, we also add a few elements to help in
our delay measurements, as explained below.

To evaluate our implementation, we focus on a firewall use case,
using a network function tree similar to that in
Figure~\ref{subfig:nf-binary-tree-unbalanced}.
A single action is applied, either the identity action, if the packet is allowed,
or marking the packet with a drop message ($0^n$), if it should be dropped.
We use three commodity PCs for our experiments (8-core Intel Xeon E5-2630 with
2.4GHz CPU and 16 GB of RAM): one for both \texttt{Entry} and
\texttt{Client}, in order to use the same clock for delay
measurements, and the other two as two \texttt{Processor}s. The four
nodes (including the two on the same machine) are connected through Intel X520
NICs, with 10-Gbps SFP+ cables. The topology is thus very similar to the one in
Figure~\ref{fig:pnfv-overview}, except that we only have $t = 2$ in $\mathcal{B}(t)$,
and that $\mathcal{A}$ and $\mathcal{C}$ share the same physical machine.
Another difference is that our machines are connected directly, without
intermediate routers between them.
We use a trace captured at one of our campus border router (pre-loaded into
memory) as input for the
\texttt{Entry} element, which executes the \texttt{Split Packet} algorithm
on a single core. Then, each output of \texttt{Entry}
(one for $\mathcal{C}$ and one per $\mathcal{B}_j$) is encapsulated
inside an UDP packet and sent to the corresponding
output device, using one core per device.

On each $\mathcal{B}_j$ machine, the packets are read from the input device,
decapsulated, and then passed to a \texttt{Processor} element which does
the actual filtering. The resulting action packets are then
re-encapsulated and sent through the NIC towards the client. This 
operation is done on a single core, but several cores can easily be used in
parallel.%
With FastClick, it suffices to launch
Click with more cores, and the system will automatically create the
corresponding number of hardware queues on the NICS, and assign a core to each
queue.
On the client side, each of the three input NICs has an associated core.
Incoming packets are decapsulated, and then passed to the
\texttt{Client} element, which reconstructs the final packets (on its
own core). Reconstructed packets which are not marked as dropped are then
passed to a receiver pipeline, which computes the entry-to-exit delay, 
counts packets and measures reception bandwidth.
To measure delays, the packets in the in-memory list are tagged with a sequence
number in the packet payload, before the transmission begins. This number
allows to match the exit timestamp with an entry timestamp, which is kept in
memory. This allows to avoid storing the timestamp itself in the packet,
which would increase the delay measured.\longver{ To store the sequence number,
 we need to extend very small packets (e.g.\ TCP ACKs). We prefer that to
 not accounting for small packets in our delay measurements.}

\longver{The SplitBox setup is compared against a simpler setup using a
 \texttt{IPFilter} element, with the same filtering rules, to act as a
 non-private firewall. In that configuration, a single machine is used. The
 \texttt{IPFilter} element replaces the \texttt{Entry} element and sends only
 the non-dropped packets (without encapsulation) directly to an output device,
 which is connected to an input device feeding the receiver pipeline.}
\shortver{The SplitBox setup is compared against a simpler setup using an
 \texttt{IPFilter} element on a single machine, to act as a non-private
 outsourced firewall with the same rules. The \texttt{IPFilter} element sends only the
 non-dropped packets (without encapsulation) directly to an output device,
 which is connected to an input device feeding the receiver pipeline.}

\section{Performance Evaluation}
\label{sec:evaluation}

We now present the results of the experiment described above,
with various input bit-rates and different number of rules,
while measuring loss rate and delays.
While we have to forward all packets to the client, a non-private outsourced
firewall can drop the rejected packets immediately.
Thus, its achievable bit-rate will depend on a combination of the input
traffic and the ruleset. To normalize results in our analysis, we craft
rulesets such that all packets are accepted. While it changes nothing for SplitBox,
it is a worst-case for the \texttt{IPFilter}-based testcase. 
At the same time, we tightly
control the number of match attempts per packet, in order to evaluate the impact
of the average number of rules traversed by a packet before it matches.

Figure~\ref{fig:bandwidth} illustrates the evolution of the maximum achievable
bandwidth (taken as inducing less than 0.001\% losses), as a function
of the number of traversed rules (i.e., the number of match attempts per
packet). Our trace packets are about 1~kB on average, so that 8~Gbps
corresponds to about 1~Mpps.
We observe that the bandwidth
decreases significantly with more traversed rules with SplitBox (PNFV),
mainly due to the hashing function, which is called on the packet header once per match attempt. 
Not only is this more computationally expensive than simpler comparisons, but it is also done
each time on different data (as we need to first XOR packet header with match
projection), taking no advantage of the cache. \longver{\texttt{IPFilter} is also
sensitive to the number of match attempts, but much less so thanks to cheaper
comparisons on a hot cache.} Fortunately, the \texttt{Processor} operation is
inherently parallelizable, thus, allocating more cores speeds things up. 
Note that the average number of traversed rules
in a real firewall is significantly lower than the total number of rules. 
Therefore, 
it is particularly important to choose the order of match attempts according to the traffic
distribution, and/or to use a more complex tree structure minimizing the number
of match attempts.

\begin{figure}[t]
 \centering
 \includegraphics[width=0.95\columnwidth]{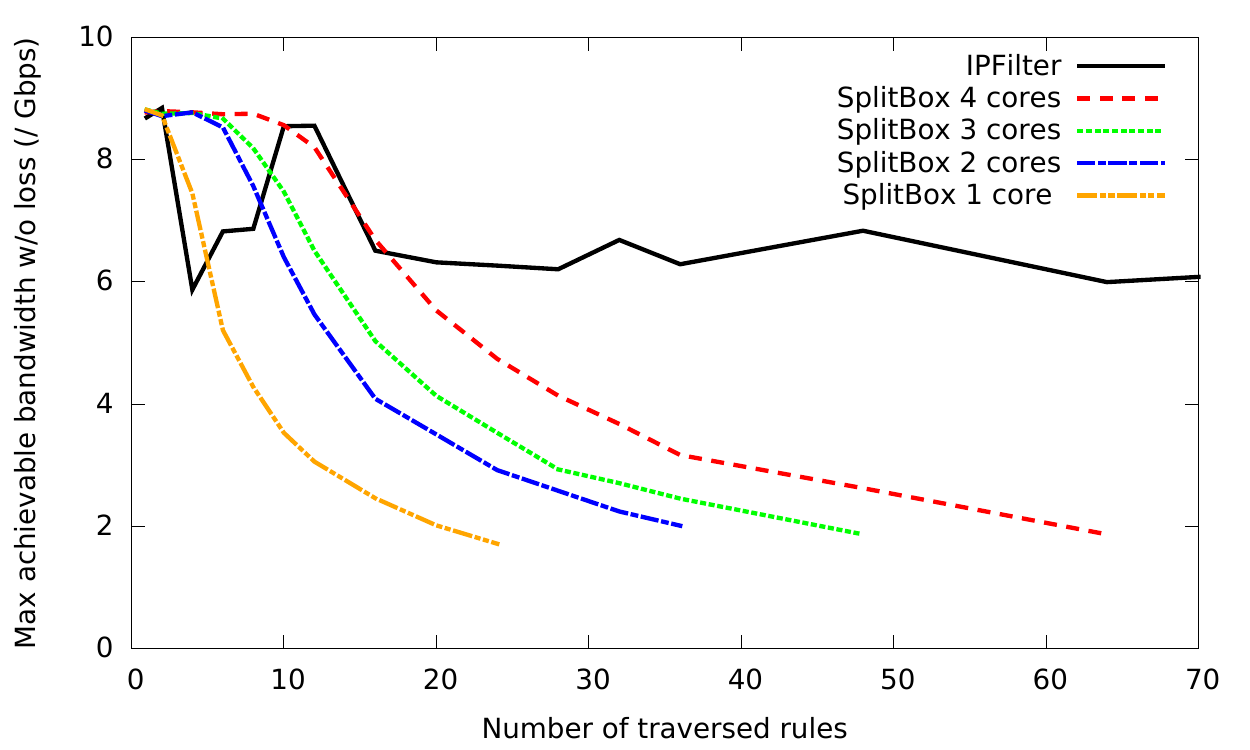}
 \ifshort\else\vspace{-0.3cm}\fi
 \caption{Achievable bandwidth drops sharply with the number of traversed
 rules.}
 \ifshort\else\vspace{-0.3cm}\fi
 \label{fig:bandwidth}
\end{figure}

\begin{figure}[t]
 \centering
 \includegraphics[width=0.95\columnwidth]{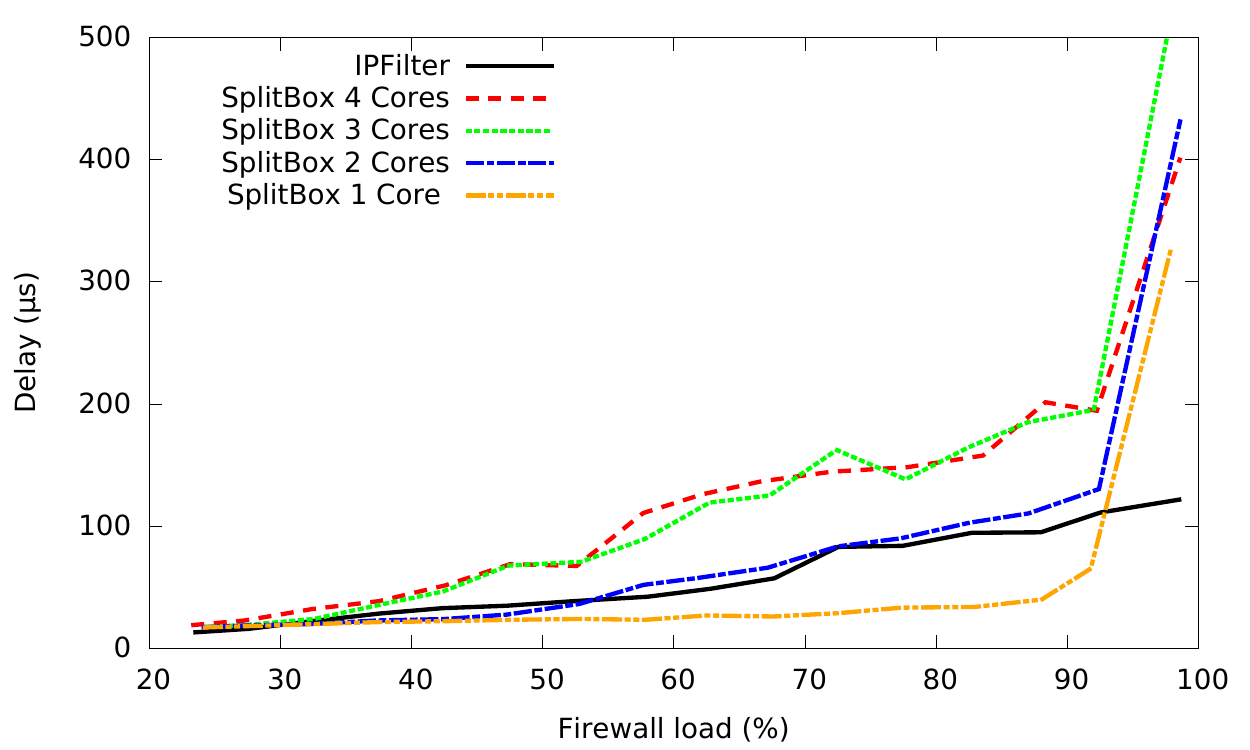}
 \ifshort\else\vspace{-0.3cm}\fi
 \caption{Delay increases with the firewall load.}
 \ifshort\else\vspace{-0.3cm}\fi
 \label{fig:delays}
\end{figure}

Finally, in Figure~\ref{fig:delays}, we plot the delays as a function of firewall load
(i.e., current input bandwidth over maximum achievable bandwidth).
Note that the delays do not follow the same dependency w.r.t.~the
number of match attempts per packet. Although these increase slightly with
the number of traversed rules, they are mostly governed by queuing delays in
the system (in NICs rings, or in-memory rings exchanging packets between
the different processing cores). 
The number of blinds $l$ seems to have little impact on the performance: with
$l$ ranging from 64 to 65,536, we observe no noticeable difference,
except for additional memory consumption.

In conclusion, our SplitBox proof-of-concept implementation 
for a firewall use case achieves comparable performance 
to a non-private version, providing
acceptable throughput and delays for small rulesets. Larger rulesets should be
carefully laid out in order to minimize the number of match attempts per
packet.

\section{Conclusion \& Future Work}
This paper presented SplitBox, a novel scalable system that allows a cloud service provider to privately compute
network functions on behalf of a client, in such a way that the cloud does not learn the network 
policies. It provides strong security guarantees in the honest-but-curious model, based on cryptographic secret sharing.
We experiment with our implementation using firewall as a test case, and achieve a throughput 
in the order of 2~Gbps, with packets of average size 1~kB traversing about 60 firewall rules. 
In future work, we plan to consider more diverse types of 
matches (that allow matching on arbitrary ranges) and actions (that allow overlapping non-zero bits), 
as well as using $k$-out-of-$t$
secret sharing schemes rather than $t$-out-of-$t$.

\bibliographystyle{abbrv}
\bibliography{bibfile}

\ifshort
\else

  \appendix
  \section{Security Proofs}
\label{app:sec-proofs}
We assume a passive (honest-but-curious) adversary $\mathcal{E}$ which can either corrupt $\mathcal{A}$, or up to $t - 1$ parties (MBs) from $\mathcal{B}(t)$.\footnote{\small We will use the word `party' instead of middlebox or MB in this section.} Let $\Pi$ denote our PNFV scheme (SplitBox). Before a formal security analysis, we first discuss the assumptions and privacy requirements of the scheme $\Pi$.
\begin{itemize}
	\item The parameter $n$ is public. 
	\item $\mathcal{A}$ should not know the network function $\psi = (M, A)$ (not even $|M|$ or $|A|$). It does however see $x$ in clear. 
	\item Each $\mathcal{B}_j \in \mathcal{B}(t)$ knows the projection $\pi_\mu$ of the match $\mu$ of each matching function $m \in M$. It should not, however, learn the match $\mu$ of any matching function $m \in M$ (beyond what is learnable through $\pi_\mu$). It also knows the result of all the of matching functions. Note that this may include matching functions that are not necessary to compute $\psi(x)$ for each packet $x$, i.e., the subset of matching functions that are in the path that exit the graph $\psi$ given $x$. Since $\mathcal{B}_j$ can always access the hash function $\mathbb{H}$ offline, it can check all matching functions $m \in M$ for their output (not necessarily in the path of $\psi$). We therefore need to make this explicit. 
	\item Each party $\mathcal{B}_j \in \mathcal{B}(t)$ should not know $x$. Furthermore, for any two packets $x_1$ and $x_2$, it should not know which bits of $x_1$ and $x_2$ are the same, beyond what is learn-able through the result of the subset of the matching functions used in $\psi(x_1)$ and $\psi(x_2)$. In particular, if a matching function $m$ has projection $\pi_{\mu}$ for its match $\mu$, it should only learn that the bits corresponding to $\pi_{\mu}$ are the same if $m(x_1) = m(x_2) = 1$. If $m(x_1) \ne m(x_2)$, $\mathcal{B}_j$ should not learn whether individual bits corresponding to $\pi_{\mu}$ are the same or different (except when $\text{wt}(\pi_\mu) = 1$). That is the reason for using the hash function $\mathbb{H}$ in the scheme. 
	\item Any coalition of $t - 1$ parties in $\mathcal{B}(t)$ should not be able to learn the action $\alpha$ and the action projection $\beta$ of every action $a \in A$. 
\end{itemize} 

Let us denote random variables $I$ and $O$ denoting the input and output of a party (or a subset of parties) corrupted by $\mathcal{E}$.\footnote{\small Somewhat abusing notation, in this section, we use the same symbol $I$ for input that was previously reserved for the identity action function.} Further denote the random variable $X$ representing the packet $x$, and $D$ representing the description of the network function $\psi$. The output of the network function $\psi$ on input from $X$ is denoted $\psi(X)$. We first describe the ideal functionality, denoted \textsc{ideal}, followed by the real setting, denoted \textsc{real}.

$\textsc{ideal}(\psi, \mathcal{S})$. We assume a trusted third party $\mathcal{T}$, which communicates with each of the parties via a secure and private link. $\mathcal{T}$ is given the network function $\psi = (M, A)$. Parties $\mathcal{B}(t)$ are given the ``index set'' of $M$ (i.e., $\{1, 2, \ldots, |M|\}$) together with the matching projections $\pi_\mu$, for the match $\mu$ of each matching function $m \in M$. Notice that, since in our protocol, we leak this information, we need to make this explicit. Party $\mathcal{A}$ receives a packet $x$ and hands it over to $\mathcal{T}$. $\mathcal{T}$ computes $x' = \psi(x)$. It hands over $x'$ to $\mathcal{C}$. Since in our protocol, we leak the information about the output of the matching functions, $\mathcal{T}$ also hands over the result of each matching function $m \in M$ to the parties $\mathcal{B}(t)$. The simulator $\mathcal{S}$ serves as the adversary in the \textsc{ideal} setting. Succinctly, $\textsc{ideal}(\psi, \mathcal{S})$ is the tuple $(I, O, X, \psi(X), D)$, where the random variables correspond to the party (or subset of parties) controlled by $\mathcal{S}$. 

$\textsc{real}(\Pi, \mathcal{E})$. Our real setting is simply the execution of our scheme in the presence of the adversary $\mathcal{E}$. It again represents the tuple $(I, O, X, \psi(X), D)$ where each random variable corresponds to the party (or subset of parties) corrupted by $\mathcal{E}$. Naturally, depending on whether $\mathcal{E}$ corrupts party $\mathcal{A}$ or upto $t-1$ parties in $\mathcal{B}(t)$, the simulator $\mathcal{S}$ in the ideal setting will be different (and so will be the random variables in the tuple $(I, O, X, \psi(X), D)$). 

With these two settings, we want to show that for every probabilistic polynomial time adversary $\mathcal{E}$ there exists a probabilistic polynomial time adversary $\mathcal{S}$, such that 
\[
\textsc{real}(\Pi, \mathcal{E})	 \approx_{\text{c}} \textsc{ideal}(\psi, \mathcal{S}),
\]
where $\approx_c$ denotes computational indistinguishability. If the above holds, we say that $\Pi$ privately processes $\psi$. In our proofs, we implicitly use the assumption that given binary strings $c$ and $c_1, \ldots, c_t$ such that $c_1, \ldots, c_{t-1}$ are random binary strings in $\{0, 1\}^n$, and $c_t = c_1 \oplus \cdots \oplus c_{t - 1} \oplus c$, then any subset of strings from $c_1, \ldots, c_t$, denoted $C(t-1)$, with cardinality $\le t - 1$, the following holds: $\mathbb{P}[c | C(t-1)] = \mathbb{P}[c] = 2^{-n}$. The proof of this assumption is standard. We use this result whenever we talk about $t$-out-of-$t$ shares in our proposed PNFV solution.

Our main results are as follows.
\begin{theorem}
The PNFV scheme $\Pi$ privately processes $\psi$ against an honest-but-curious $\mathcal{E} = \mathcal{A}$.
\end{theorem}
\begin{proof}
Before receiving any packet, the simulator $\mathcal{S}$ samples $l$ uniformly random strings $s_i \in \{0, 1\}^n$ to construct the lookup table $S$ and gives it to $\mathcal{E}$. It initializes its counter to $0$. Upon receiving a packet $x$, $\mathcal{S}$ forwards it to $\mathcal{T}$. For $\mathcal{E}$, $\mathcal{S}$ first gets the current value of the counter $i \in [l]$. It further samples a uniformly random $r \in \{0, 1\}^n$ and constructs $x_{\mathtt{w}} \leftarrow x \oplus r$. It computes $t$ shares of $r$, the $j$th share of which is denoted $r_j$. Finally it obtains $x_\mathtt{r} \leftarrow x(1, n) \oplus s_i$ by looking up the counter value $i$ in the table $S$. Finally $\mathcal{S}$ gives $x_{\mathtt{r}}$, $i$, $x_{\mathtt{w}}$ and the $t$ shares of $r$ to $\mathcal{E}$. Once the counter $i$ reaches $l$, $\mathcal{S}$ resets it to $0$. 

Since the input to party $\mathcal{A}$ is the same as the input packet $x$, we have that $I = X$ (which holds both in the ideal and real setting). The output $O$ is distributed in the exact same manner in the two worlds. Since the output is generated without any knowledge of the network function $\psi$, we have that $D$ is the same in the ideal and real world. Finally, the output of $\psi$ is not revealed in the two worlds. Hence 
$\textsc{real}(\Pi, \mathcal{E}) =\vspace{-0.1cm}$ $\textsc{ideal}(\psi, \mathcal{S}) \Rightarrow \textsc{real}(\Pi, \mathcal{E})	 \approx_{\text{c}} \textsc{ideal}(\psi, \mathcal{S})$. 
\end{proof}

As discussed in Section~\ref{sub:security}, if the match of a matching function is small, the adversary can brute-force the hash function $\mathbb{H}$ to find its pre-image. Thus, our security proof for $\mathcal{E} \subset \mathcal{B}(t)$ requires that the minimum Hamming weight of a match $\mu$ in the set of matching functions $M$ should be large enough for brute-force to be infeasible. Furthermore, our security proof applies only when the blinds are used once, i.e., for counter values $\le l$ without reset. See Section~\ref{sub:security} for our proposed mitigation strategy for security, when the counter completes its cycle.
\begin{theorem}
Suppose $\delta = \min_\mu \text{wt}(\pi_\mu)$, for all matching functions $m \in M$. The PNFV scheme $\Pi$ privately processes $\psi$, up to $l$ inputs (packets), against an honest-but-curious $\mathcal{E} \subset \mathcal{B}(t)$ in the random oracle model.
\end{theorem}
\begin{proof}
Let $\mathcal{R}: \{0, 1\}^* \rightarrow \{0, 1\}^{q}$ denote the random oracle. Before receiving any packet, the simulator $\mathcal{S}$ simulates the lookup table $\tilde{S}$ as follows. For each $m \in M$, given the projection $\pi_\mu$ of its match $\mu$, it generates $l$ binary strings by sampling a random bit where $\pi_\mu(i) = 1$ and placing a $0$ otherwise. For each such string, $\mathcal{S}$ samples a uniform random binary string of length $q$. $\mathcal{S}$ creates two tables. One is the lookup table $\tilde{S}$, and the other its personal table $\hat{S}$. The table $\hat{S}$ contains the pre-images of the entries in $\tilde{S}$. It hands over $\tilde{S}$ to each party in $\mathcal{E}$. For each policy $(m, a) \in \psi$, it generates $|\mathcal{E}|$ random binary strings $\alpha_j$ and $\beta_j$ of length $n$, for $1 \le j \le |\mathcal{E}|$, and gives each pair $(\alpha_j, \beta_j)$ to a separate player in $\mathcal{E}$. $\mathcal{S}$ initiates a counter $i$ initially set to $0$.

Upon receiving the result of the matching functions in $M$ from $\mathcal{T}$, indicating the arrival of a new packet, $\mathcal{S}$ first generates a random binary string as $x_{\mathtt{w}}$ and $|\mathcal{E}|$ random binary strings of length $n$ (to simulate the $r_j$'s). $\mathcal{S}$ initializes an empty string $x_{\mathtt{r}}$. For each matching function $m$ that outputs 1, $\mathcal{S}$ looks up its table $\hat{S}$ and the projection $\pi_{\mu}$, where $\mu$ is the match of the matching function, and replaces the corresponding bits of $x_{\mathtt{r}}$ with the corresponding bits of the input string to the lookup table $\hat{S}$. Finally, for all bits of $x_{\mathtt{r}}$ that are not set, $\mathcal{S}$ replaces them with uniform random bits. It hands over $x_{\mathtt{w}}$, $x_{\mathtt{r}}$ and $r_j$ to each party in $\mathcal{E}$, together with the current counter value $i$. 

For any oracle query from a party $\mathcal{B}_j \in \mathcal{E}$, $\mathcal{S}$ first looks at its table $\hat{S}$ and sees if an entry exists. If an entry exists, $\mathcal{S}$ outputs the corresponding output from the table $\hat{S}$. If an entry does not exist, $\mathcal{E}$ outputs a uniform random string of length $q$, and stores the input and the output by appending it to the table $\hat{S}$. 

It is easy to see that the distribution of the variables $(I, O, X, \psi(X), D)$ for each party in $\mathcal{E}$ is the same as in the real setting, for any $\mathcal{E}$, such that $|\mathcal{E}| < t$, for any value of the counter $i \le l$, and for a polynomial in $\delta$ number of oracle queries. Therefore $\textsc{real}(\Pi, \mathcal{E})	 \approx_{\text{c}} \textsc{ideal}(\psi, \mathcal{S})$.
\end{proof}

\fi

\end{document}